\documentclass[12pt,a4paper]{article}
\usepackage[a4paper,margin=1in]{geometry}
\usepackage[dvips]{graphicx}
\usepackage{color}
\usepackage{cite}
\usepackage{hyperref}
\usepackage[cmex10]{amsmath}
\usepackage{amsthm}
\usepackage{amssymb}
\usepackage{algorithmic}
\usepackage{url}
\usepackage{algorithm}
\usepackage{subfigure}
\usepackage{wrapfig}

\newtheorem{theorem}{Theorem}[section]
\newtheorem{lemma}[theorem]{Lemma}
\newtheorem{corollary}[theorem]{Corollary}

\title{Attractive Ellipsoid Sliding Mode Observer Design for State of Charge Estimation of Lithium-ion Cells}

\author{	Anirudh Nath\emph{},
	Raghvendra Gupta\emph{},
	Rohit Mehta\emph{},
	Supreet Singh Bahga\emph{},\\
	Amit Gupta\emph{}\\ 
	Department of Mechanical Engineering\\
Indian Institute of Technology Delhi\\
Hauz Khas, New Delhi 110016, India.\\
Email: anirudh.nath88@gmail.com;raghvendrag6@gmail.com;\\rohit0149@gmail.com;bahga@mech.iitd.ac.in;agupta@mech.iitd.ac.in\\
and\\
Shubhendu Bhasin\emph{},\\
Department of Electrical Engineering\\
Indian Institute of Technology Delhi\\
Hauz Khas, New Delhi 110016, India.\\
Email: sbhasin@ee.iitd.ac.in
}
\date{}
\begin{document}
\maketitle

\begin{abstract}
This work investigates the real-time estimation of the state-of-charge (SoC) of Lithium-ion (Li-ion) cells for reliable, safe and efficient utilization. A novel attractive ellipsoid based sliding-mode observer (AESMO) algorithm is designed to estimate the SoC in real-time. The algorithm utilizes standard equivalent circuit model of a Li-ion cell and provides reliable and efficient SoC estimate in the presence of bounded uncertainties in the battery parameters as well as exogenous disturbances. The theoretical framework of the observer design is not limited to the SoC estimation problem of Li-ion cell but applicable to a wider class of nonlinear systems with both matched and mismatched uncertainties. The main advantage of the proposed observer is to provide a fast and optimal SoC estimate based on minimization over the uncertainty bound. The proposed method is experimentally tested and evaluated using the hybrid pulse power characterization test (HPPC) and urban dynamometer driving schedule (UDDS) test data, which demonstrate its effectiveness and feasibility.
\end{abstract}

\section{Introduction}

Lithium-ion (Li-ion) cells are ubiquitous energy storage sources which provide a promising solution to the global future energy needs. In comparison to other battery technologies \cite{3,4}, Li-ion cells provide several advantages such as excellent energy-to-weight ratio, no memory effect, and low self-discharge rate in an unused state. All these favourable characteristics in conjunction with rapidly reducing costs have established Li-ion cells as the indispensable component for a wide variety of applications in the energy sector, especially in automotive, smart-grid and aerospace industries \cite{1,2}.\par

Important issues associated with the use of Li-ion cells, including reliability, efficiency, and longevity, demand an efficient battery management system (BMS) capable of monitoring critical internal states of Li-ion cells such as the state of charge (SoC), state of health, and energy density, etc. \cite{1,5,6}. To be specific, the SoC is a vital indicator of the actual amount of usable charge and energy content of the Li-ion cell under operation, and requires monitoring to control the extent of charging and discharging to avoid overcharge or over-discharge. Due to nonlinear physics, it is not possible to directly measure SoC using external electrical signals and thus needs to be estimated. The simplest method for estimating the SoC is the ampere-hour counting, an open-loop technique that requires precise knowledge of the initial SoC, which is typically not available and a poor initial guess often leads to accumulation of errors \cite{10}. Another important method of SoC estimation, the open-circuit voltage (OCV) method is an uncomplicated procedure but not useful for online computation since it requires a long relaxation time for accurate and precise measurement of the OCV. Black-box estimation methods \cite{7,8,9} have drawbacks like overfitting, extensive training, difficult online adaptation, and high computational cost \cite{10}.\par

\begin{figure*}
\centering
	\includegraphics[width=15cm]{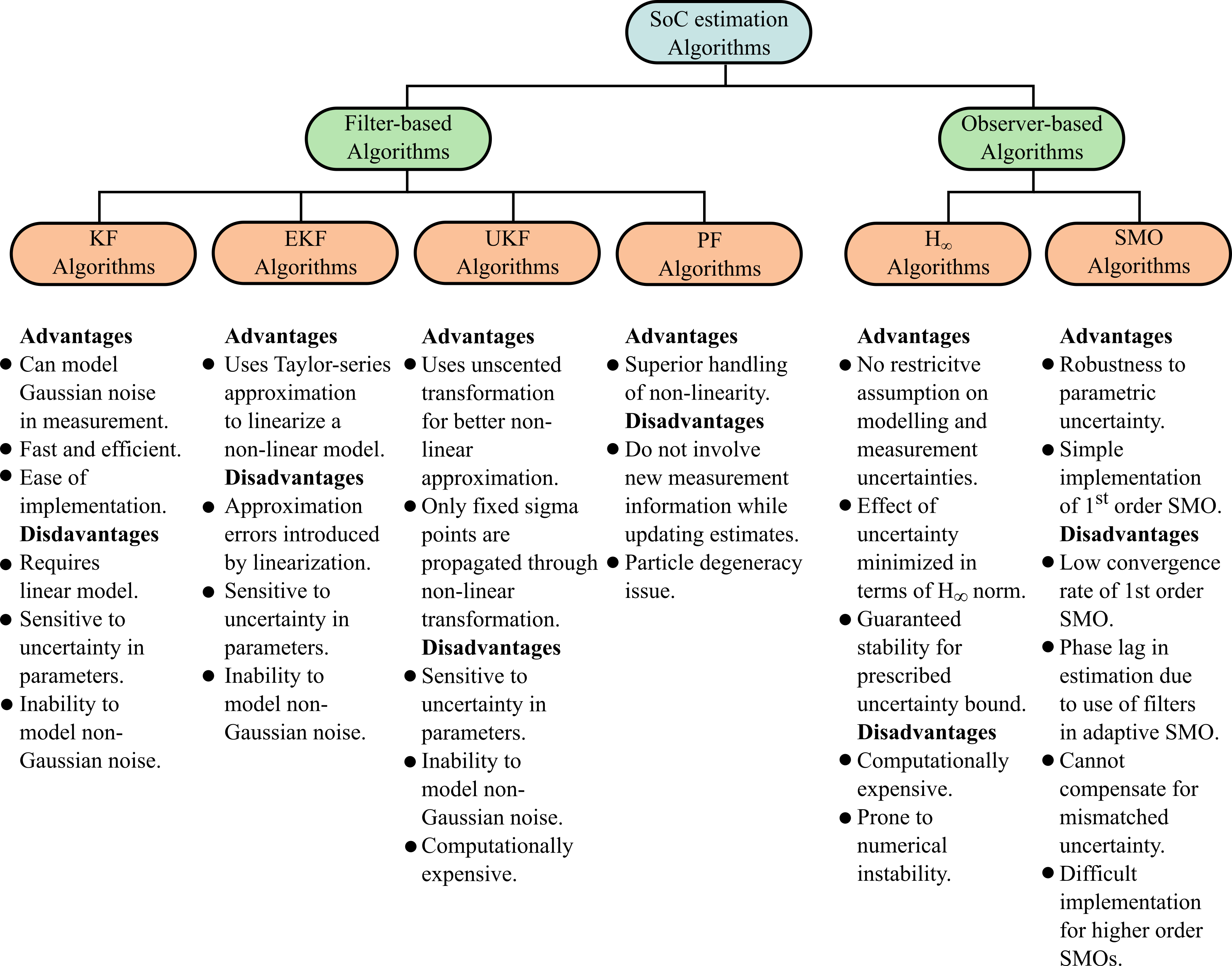}
	\caption{Various filter-based and observer-based SoC estimation approaches.}\label{FIG_demerit}
\end{figure*}

Model-based SoC estimation techniques can be broadly classified into two main classes, namely, the filter-based estimation and the observer-based estimation techniques \cite{19}. A summary of the various model-based SoC estimation algorithms along with their limitations is pictorially depicted in Fig. \ref{FIG_demerit}. KF based approaches have been extensively applied for SoC estimation \cite{11}. The requirement of linear input-output relation leads to the nonlinear characteristic of the SoC-voltage relationship being lost, thereby limiting the effectiveness of this algorithm \cite{12,13}.  The use of Taylor series in the extended Kalman filter (EKF) based approaches to linearize the nonlinear characteristics may often result in erroneous SoC estimation \cite{13,14,15}. Linearization is avoided in unscented Kalman filter (UKF) based approaches \cite{16,19}. However, UKF based methods are sensitive to inaccuracies in the initial conditions and unknown disturbances \cite{16}. The particle filtering (PF) methods can more effectively deal with the characteristics of the Li-ion cell but also have certain limitations as provided in Fig. \ref{FIG_demerit} \cite{11,16,17,18}.\par

Observer-based techniques can address limitations associated with the aforementioned filter-based SoC algorithms. To mention some of the most important works,  one may refer to SoC estimation based on reduced-order observer technique \cite{20}, model reference adaptive observer \cite{21}, proportional-integral observer \cite{46}, unknown input observer \cite{47} and the nonlinear observers designed in \cite{22,23}. Several observer-based algorithms are complicated due to the augmentation of online parameter identification algorithms based on SoC \cite{20} that can be affected by external factors such as temperature, thermal parameters and ageing effects \cite{11}. Robust observer-based SoC algorithms that utilize simplistic mathematical models of the Li-ion cell and information about the uncertainty bounds are designed to address these issues. The most important variants of robust observer-based SoC estimation schemes in the literature are based on $H_{\infty }$ filtering technique \cite{27,28} and sliding mode observer (SMO) methods \cite{30}.  As summarized in Fig. \ref{FIG_demerit}, the $H_{\infty }$ filter-based SoC estimation algorithms are robust to modelling inaccuracies and bounded disturbances, but demand high computational power and implementation cost \cite{11}.\par

The SMO-based algorithms can be classified as (i) constant \cite{31,35,36,39} and adaptive (time-varying) switching gain \cite{31,34} and (ii) first-order \cite{33,34,35,36} or second-order (based on the order of state equations) \cite{37,38}, as discussed in \cite{31}. While first-order SMOs are simple in their design and implementation \cite{33,34,35} the second-order SMOs have been shown to have higher accuracy in SoC estimation \cite{37,38}. Similarly, an adaptive SMO with time-varying switching gain provides superior performance in SoC estimation as compared to their constant gain counterpart. However, in addition to the implementation of this algorithm being difficult, the use of low-pass filters further adds to the complexity and the implementation cost \cite{31}. Furthermore, the adaptive SMO based SoC estimation can be poor due to the phase lag introduced by the use of filters. \par

The attractive ellipsoid method presents an efficient robust control strategy based on the invariant ellipsoid method in the numerically efficient linear matrix inequality (LMI) framework \cite{40,411,42,43}. The main motivation behind the current work is to propose a robust observer-based SoC estimation algorithm that can address the problems of the existing filter-based and observer-based approaches. Here a novel first order SMO is proposed which utilizes a simple equivalent circuit model (ECM) of the Li-ion cells with fast convergence and avoids the limitations of existing adaptive SMOs as discussed above. The contributions of the proposed observer-based algorithm is summarized below. The main theoretical contribution of the proposed attractive ellipsoid sliding mode observer (AESMO) is that this design applies to a wide class of uncertain nonlinear affine systems with both matched and mismatched uncertainty, including bounded exogenous disturbances. This method ensures a guaranteed convergence of SoC estimation error trajectories to a bounded ellipsoid of a minimal size where the observer gain matrix is obtained by solving a convex optimization problem with linear matrix inequality (LMI) constraints. It also avoids requirement of high-end computational resources for its implementation unlike the $H_{\infty }$ and Kalman filter-based algorithms as discussed before. Another important feature of this design is that the rate of convergence of the estimated trajectories can be altered by adjusting the value of a design parameter. Thus the issue of slow convergence of existing first-order SMOs can be improved.

\section{MATHEMATICAL MODELLING OF LI-ION CELL}
In this section, the overall framework of the mathematical model of Li-ion cell, parameter identification and the SoC estimation technique will be presented.  

\subsection{Equivalent Circuit Model of Li-ion cell}

\begin{figure}
\centering
	\includegraphics[width=8.5cm]{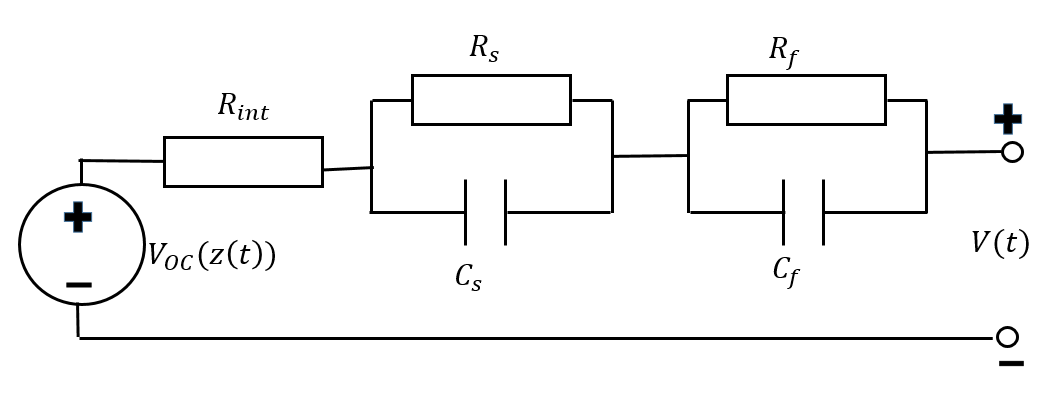}
	\caption{Block diagram representation of the equivalent circuit model of Li-ion cell.}\label{FIG_1}
\end{figure}

Here, an equivalent circuit model (ECM) with dual polarization RC circuit is considered for the purpose of SOC estimation. The ECM provides a good trade-off between complexities and precision of the electrical behaviour of the Li-ion cell \cite{1}. Various versions of ECM models exist in the literature as mentioned in \cite{2,15,25,26,27,30}. The presented observer design can accomodate modelling errors and parametric uncertainties. A schematic of the Li-ion ECM is shown in Fig. \ref{FIG_1}, where $R_{int}$ is the internal ohmic resistance of the Li-ion cell which accounts for the energy losses of the Li-ion cell during its operation. The RC networks $(R_{s},C_{s})$ and $(R_{f},C_{f})$ are, respectively, used to describe the short-term as well as long-term transient behaviour of the Li-ion cell. The mathematical relationship of the terminal voltage of the Li-ion cell which is expressed as linear sum of the OCV ($V_{\mathrm{oc}}(\mathrm{z(t)})$), the voltage drops across the internal resistance as well as the RC-modules (i.e., $R_{int} I(t)$, $V_{RC_1}(t)$ and $V_{RC_2}(t)$, respectively) are given as
\begin{equation}\label{eq1}
V(t)=V_{\mathrm{oc}}(\mathrm{z(t)})-V_{RC_1}(t)-V_{RC_2}(t)-R_{int} I(t)
\end{equation}
The voltage drops across the RC modules representing the slow and fast polarization characteristics of the Li-ion cell, respectively, are given by 
\begin{equation}\label{eq2}
\dot{V}_{RC_1}(t)=-\frac{1}{R_{s} C_{s}} V_{RC_1}(t)+\frac{1}{C_{s}} I(t)
\end{equation}
\begin{equation}\label{eq3}
\dot{V}_{RC_2}(t)=-\frac{1}{R_{f} C_{f}} V_{RC_2}(t)+\frac{1}{C_{f}} I(t)
\end{equation}
The SoC of Li-ion cell, $z(t)$ is related to the current and the nominal capacity of the Li-ion cell as follows:
\begin{equation}\label{eq4}
\dot{z}(t) =-\frac{1}{Q} I(t) 
\end{equation}
where $Q$ denotes the total capacity of the cell. The SoC can be expressed in terms of the terminal voltages and other voltage drops by substituting the expression for the current, $I(t)$ from \eqref{eq1} in \eqref{eq4} as
\begin{equation}\label{eq5}
\dot{z}(t) =-\frac{1}{R_{int} Q}\left(V_{\mathrm{OC}}(z(t))-V_{RC_1}(t)-V_{RC_2}(t)-V(t)\right)
\end{equation}
Differentiating the terminal voltage, $V(t)$ with respect to time and assuming negligible change in current in between the sampling instants (i.e. $\frac{dI(t)}{dt}\simeq 0$) as in \cite{11} and further utilizing \eqref{eq2} and \eqref{eq3}, one can obtain
\begin{multline}\label{eq6}
\begin{aligned}
\dot{V}(t)=& \dot{z} \frac{\partial{V}_{\mathrm{OC}}(\mathrm{z})}{\partial z}-\dot V_{RC_2}-\dot V_{RC_1}
-R_{int} \frac{dI}{dt}\\
=&-\frac{I}{Q} \frac{\partial{V}_{\mathrm{OC}}(\mathrm{z})}{\partial z} +\frac{1}{R_{f} C_{f}} V_{RC_2}-\frac{I}{C_{f}} \\
&+\frac{1}{R_{s} C_{s}} V_{RC_1}-\frac{I}{C_{s}} 
\end{aligned}
\end{multline}
Now all the dynamical equations described above can be summarised as follows
\begin{equation}\label{eq7}
\begin{aligned}
\dot{V} =& a_2 V_{OC}(z)-a_2 V + (a_3-a_2) V_{RC_2} \\
&-\left[b_{1} \times  \frac{d{V}_{\mathrm{OC}}(\mathrm{z})}{dz}+b_{2}+b_{3}-a_2R_{int}\right] I \\
\dot{z}=&- b_1 \bar R \left(V_{\mathrm{OC}}(\mathrm{z})-V-V_{RC_1}-V_{RC_2}\right) \\
\dot{V}_{RC_1}=&-a_{2} V_{RC_1}+b_{2} I\\
\dot{V}_{RC_2}=&-a_{3} V_{RC_2}+b_{3} I
\end{aligned}
\end{equation}
where $a_{2}\triangleq \frac{1}{R_{s} C_{s}},  a_{3}\triangleq\frac{1}{R_{f} C_{f}},
$
$b_{1}\triangleq\frac{1}{Q}, b_{2}\triangleq\frac{1}{C_{s}}$, $  b_{3}\triangleq\frac{1}{C_{f}}$ and $\bar R\triangleq\frac{1}{R_{int}}$. The dynamical equations in \eqref{eq7} can be expressed as
\begin{equation}\label{sys}
    \dot x(t)=(A+\Delta A(t))x(t)+Bu(t)+\phi(x(t),u(t))+Dd(t)
\end{equation}
and the output equation as
\begin{equation}\label{op}
    y(t)=Cx(t),~~~y\in \mathbb{R}
\end{equation}
where the state $x(t)\triangleq[V(t)~z(t)~V_{RC_1}(t)~V_{RC_2}(t)]^T\in \mathbb{R}^4$, input $u(t) \triangleq I(t)\in \mathbb{R}$, exogenous disturbance $d(t)\in \mathbb{R}$, $A\triangleq \left[\begin{array}{cccc}
-a_{2} & a_{2} \alpha_{1} & 0 & a_3-a_2\\
-b_1 \bar R & -\alpha_{1}b_1 \bar R & -b_1 \bar R & -b_1 \bar R \\
0 & 0 & -a_{2} & 0\\
0 & 0 & 0 & -a_3
\end{array}\right]$, 
$B \triangleq \left[\begin{array}{c}
-b_{2}  -b_{3} +a_2 R_{int}\\
0 \\
b_2\\
b_3
\end{array}\right]$,\\
$C=[1~0~0~0]$,
$\phi \triangleq \left[\begin{array}{c}
a_2 V_{OC}(z)-a_{2} \alpha_{1} z-b_1\frac{\partial{V}_{\mathrm{OC}}(\mathrm{z})}{\partial z}I\\
-b_1 \bar R V_{OC}(z)+\alpha_1 b_1 \bar Rz \\
0\\
0
\end{array}\right]$ and $D\triangleq [1~1~1~1]^T$. The system matrix in \eqref{sys} is partitioned into a nominal matrix, $A$ and an uncertain  matrix, 
$\Delta A(t)=\left[\begin{array}{cccc}
-\Delta a_{2} & \Delta a_{2} \alpha_{1} & 0 & \Delta a_3-\Delta a_2\\
-b_1 \Delta {\bar R} & - \alpha_{1}b_1 \Delta {\bar R} & -b_1 \Delta {\bar R} & -b_1 \Delta {\bar R} \\
0 & 0 & -\Delta a_{2} & 0\\
0 & 0 & 0 & -\Delta a_3
\end{array}\right]$. 
The elements of $\Delta A\left( t\right)$ vary within a specified interval $a_2\in[a_{2_{min}},a_{2_{max}}]$, $a_3\in[a_{3_{min}},a_{3_{max}}]$ and $\Delta {\bar R}\in[\Delta {\bar R}_{min},\Delta {\bar R}_{max}]$, where $a_{2_{min}}$, $a_{3_{min}}$, $\Delta {\bar R}_{min}$ and $a_{2_{max}}$, $a_{3_{max}}$, $\Delta {\bar R}_{max}$, respectively, are the known minimum and maximum values of the uncertain parameters.

\subsection{Parameter estimation}
In the present work, the influence of ambient temperature, humidity, etc., has been neglected. Hence, it is assumed that the parametric values of $R_{int}$, $R_{s}$, $C_{s}$, $R_{f}$ and $C_{f}$ in the ECM can be considered as dependent on SoC only. For the sake of simplicity, the above ECM parameters are assumed to be independent of the current direction (i.e. any hysteresis phenomenon during charging and discharging process is ignored).\par
\begin{figure}
\centering
	\includegraphics[width=9cm,height=4.5cm]{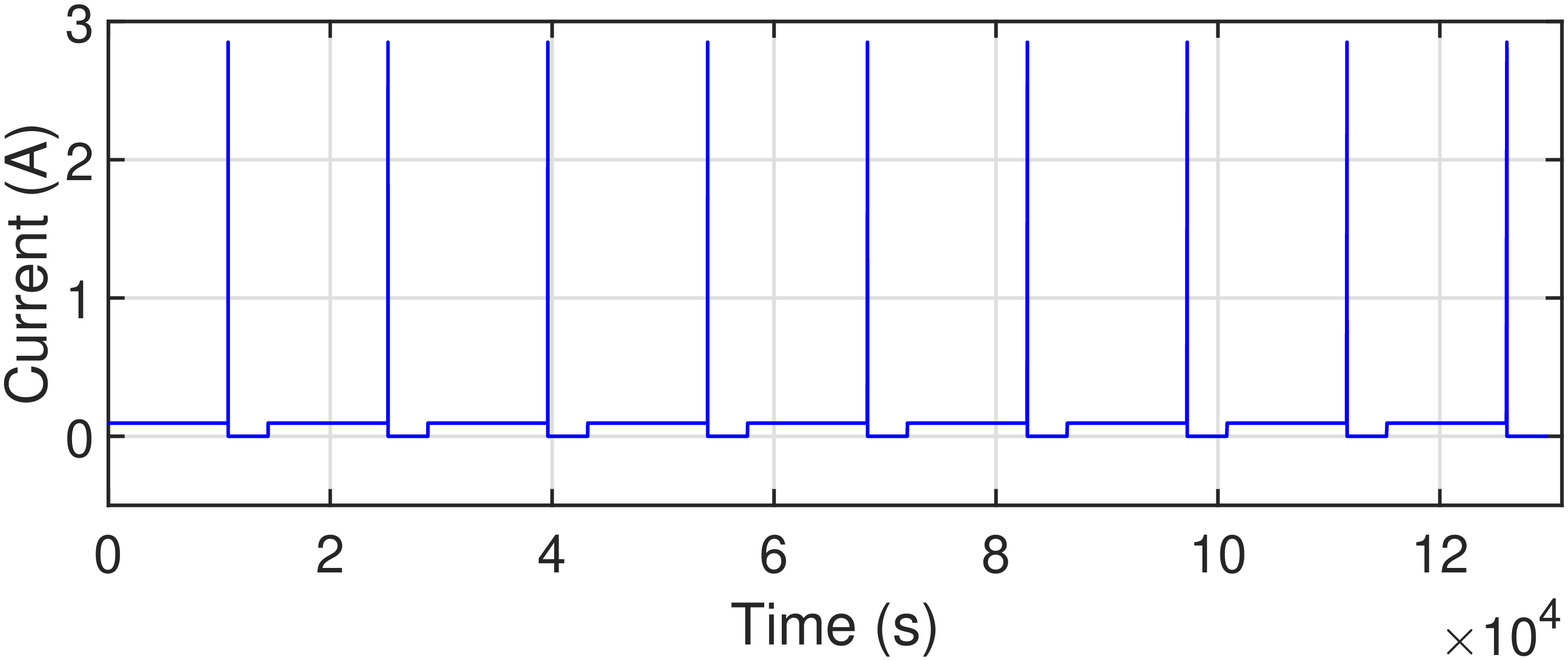}
	\caption{Current profile for pulse discharge test utilized for parameter identification.}
	\label{FIG_hppcident}
\end{figure}
To calculate parameters at a particular SoC accurately, the discharging current profile employed is shown in Fig. \ref{FIG_hppcident}. For this purpose, Samsung  INR18650-29E Li-ion cells with nominal capacity of $2.85$ Ah and Arbin  LBT21084 Battery Cycler are used. The experimental current profile consists of a $100$ mA discharge current from $100\%$ SoC (fully charged condition) till the time the SoC is $90\%$. After this, a pulse of $2.85$ A current is applied for $10$ s followed by rest of $4410$ s. This cycle is repeated for each step until the SoC drops to $10\%$ as depicted in Fig. \ref{FIG_hppcident}. The parameter identification procedure is adopted from \cite{45} that utilizes nonlinear least squares algorithm to  determine the parameters, $R_{int}$, $R_{s}$, $C_{s}$, $R_{f}$ and $C_{f}$ of the ECM using the experimental data of the applied current and corresponding terminal voltage of the Li-ion cell. The results of the estimated ECM parameters are provided in Table \ref{tab1}. \par
\begin{table}[!h]
\renewcommand{\arraystretch}{0.1}
\centering \caption{Parameter identification.}\vspace{0.1cm}
\label{tab1} 
\begin{tabular}{cccccc}
\hline 
\hline
\vspace{0.1cm}
SoC & $R_{int}~(m\Omega)$ & $R_{s}~(m\Omega)$ & $C_{s}~(kF)$ & $R_{f}~(m\Omega)$ & $C_{f}~(kF)$\\ \hline \vspace{0.1cm}
$0.9$ & $32.00$ & $15.30$ & $1.935$ & $32.00$ & $16.78$\\ \vspace{0.1cm}
$0.8$ & $32.90$ & $23.00$ & $1.425$ & $20.40$ & $14.81$\\ \vspace{0.1cm}
$0.7$ & $30.20$ & $25.60$ & $1.401$ & $22.90$ & $10.86$\\ \vspace{0.1cm}
$0.6$ & $30.60$ & $25.60$ & $1.541$ & $71.10$ & $3.89$\\ \vspace{0.1cm}
$0.5$ & $30.60$ & $14.40$ & $2.031$ & $19.30$ & $16.43$\\ \vspace{0.1cm}
$0.4$ & $32.00$ & $15.10$ & $2.114$ & $25.20$ & $11.82$\\ \vspace{0.1cm}
$0.3$ & $30.80$ & $14.40$ & $2.419$ & $57.00$ & $65.54$\\ \vspace{0.1cm}
$0.2$ & $32.10$ & $14.90$ & $2.084$ & $23.10$ & $11.25$\\ \vspace{0.1cm}
$0.1$ & $35.50$ & $18.20$ & $1.601$ & $69.00$ & $01.36$\\ 
\hline
\hline
\end{tabular}
\end{table}
An accurate mathematical relationship between the OCV and SoC is vital for the estimation accuracy since it captures the nonlinear dynamics of the Li-ion cell. The SoC versus OCV relationship was determined using the voltage versus time data obtained by discharging the cell from $4.2$ V to $3$ V at C/$30$ rate. The initial state of the cell was ensured to be $4.2$ V by charging it up to this value at a current of  C/$30$, followed by a constant voltage charging until the current dropped below  C/$100$. The SoC values at $4.2$ V and $3$ V were taken to be $100\%$ and $0\%$, respectively, and the intermediate SoC values were determined by coulomb counting method. In this paper, a $9^{th}$-order polynomial function is considered for representing the nonlinear relationship between the OCV and SoC using the least-squares technique \cite{15} as
\begin{multline}
    V_{oc}(z) = p_1 z^9 + p_2 z^8 + p_3 z^7 + p_4 z^6 + p_5 z^5+ p_6 z^4\\
   + p_7 z^3 +p_8 z^2+p_9 z+p_{10}
\end{multline}
where the coefficients of the function are identified as $p_1=1.937\times 10^{3}$, $p_2=-8.962\times 10^{3}$, $p_3=1.745\times 10^{4}$, $p_4=-1.860\times 10^{4}$, $p_5=1.177\times 10^{4}$, $p_6=-4.514\times 10^{3}$, $p_7=1.028\times 10^{3}$, $p_8=-133.501$, $p_9=10.0891$ and $p_{10}=3.043$. 
 The result of the polynomial function identification is illustrated in Fig. \ref{FIG_ocv} which shows a good fit with the experimental data. 
\begin{figure}
\centering
	\includegraphics[width=8.5cm,height=4.5cm]{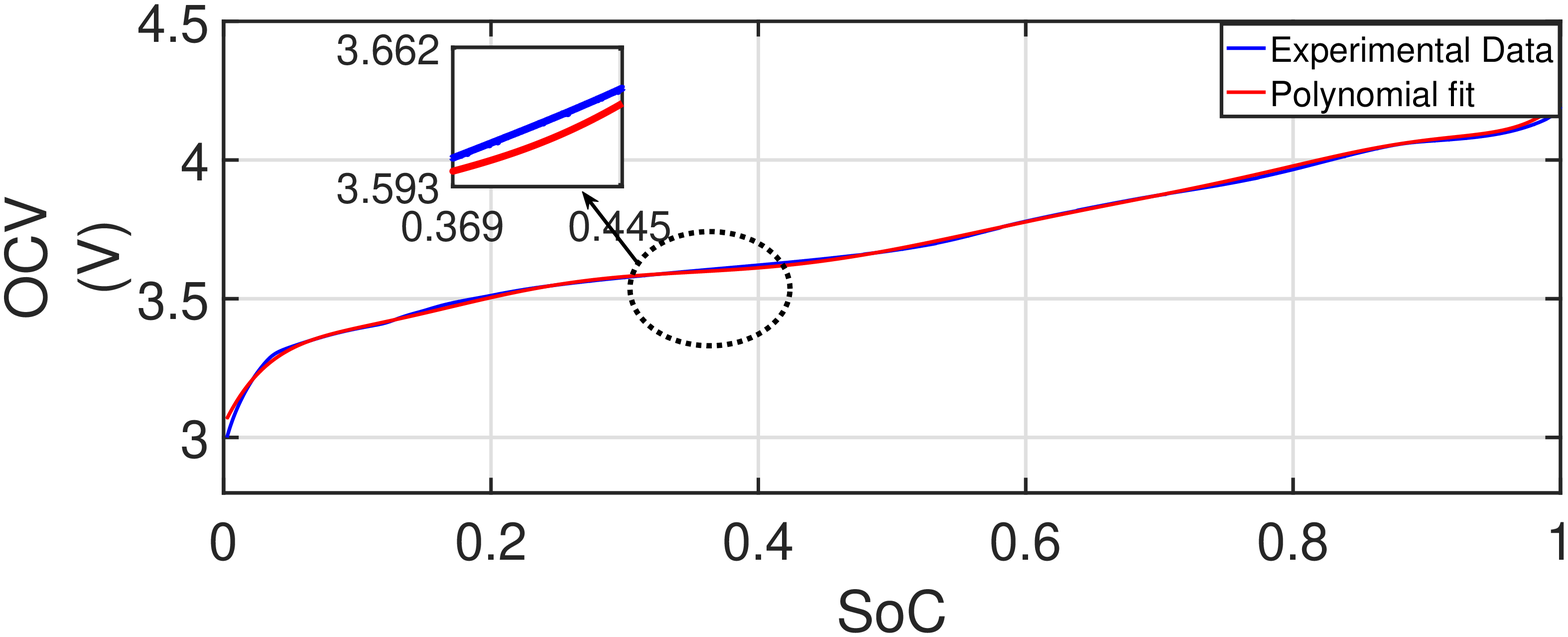}
	\caption{Open-circuit voltage vs. SoC of a Li-ion cell.}\label{FIG_ocv}
\end{figure}
\section{DESIGN OF THE ATTRACTIVE ELLIPSOID BASED SMO OBSERVER}
In this section, a novel SMO design algorithm for a class of uncertain nonlinear Lipschitz systems as provided in \eqref{sys} is presented. The design procedure involves an appropriate selection of a sliding surface in the framework of the invariant ellipsoid method \cite{43}. The proposed observer ensures that the state estimation error trajectories are ultimately confined within an attractive ellipsoid around the sliding manifold. The following assumptions are considered for the observer design procedure. \par 
Let us consider an observer with Luenberger's sliding mode structure (with the additional signum-term) as
\begin{equation}\label{ob}
\dot{\hat x}(t)=A \hat x(t)+Bu(t)+\phi(\hat x(t))+L\sigma(t)+L_s SIGN(\sigma(t))
\end{equation}
where the output error $\sigma(t)\triangleq y(t)-C\hat x(t)=C \tilde x(t) \in \mathbb{R}^m$, the observer gain matrices, $L,L_s\in \mathbb{R}^{n\times m}$, state estimate $\hat x(t)\in \mathbb{R}^n$ and the state estimation error $\tilde x(t)\triangleq x(t)-\hat x(t)\in \mathbb{R}^n$, $SIGN(\sigma)=[sign(\sigma_1),sign(\sigma_2),...,sign(\sigma_m)]^T \in \mathbb{R}^m$ and the signum function $sign(\sigma_i)$ is defined as,
\begin{equation*}
sign(\sigma_i) = \begin{cases}
1 &\text{if $\sigma_i>0$}\\
-1 &\text{if $\sigma_i<0$}\\
\in[-1,1] &\text{if $\sigma_i=0$}\\
\end{cases}
i=1,2,\dots,m
\end{equation*}
and $\sigma \in \mathbb{R}^{m}$ defines the sliding surface. Using \eqref{sys} and \eqref{ob}, the closed-loop error dynamics can be written as
\begin{multline}\label{err}
    \dot{\tilde x}(t)=A \tilde x(t)+\Delta A(t) x(t)+\Delta \phi -L\sigma(t)-L_s SIGN(\sigma(t))\\
    +Dd(t)
\end{multline}
where $\Delta \phi\triangleq \phi(x(t))-\phi(\hat x(t))$. 

The following definition, lemma and assumptions facilitate the subsequent design of a robust observer for the system in (\ref{eq7}) in an LMI framework.
\begin{lemma} 
A nonlinear mapping $\phi(x,u)$ is called Lipschitz function, if the following mathematical condition is satisfied,
\begin{equation}\label{lip}
\|\phi(x,u)-\phi(\hat{x},u)\|\leq L_{\phi}\|x-\hat{x}\|
\end{equation}
for any $(x,\hat{x}) \in \Re^{n}$ and $L_{\phi}>0$ in (\ref{lip}) denotes a Lipschitz constant \cite{411}.
\end{lemma}
\textit{Assumption 1:} The states of \eqref{sys} satisfy the mathematical inequality, $\|x(t)\| \leq X_{+}$ where $X_{+}$ is a known positive constant. \newline
\textit{Assumption 2:}The uncertain matrix is bounded, $\|\Delta A(t)\|\leq \gamma$. \newline
\textit{Assumption 3:} The exogenous disturbance is also bounded as $\|d(t)\| \leq D_{+}$ where $D_{+}$ is a known positive constant.  \newline
\textit{Assumption 4:} $\phi(x,u)$ in \eqref{sys} satisfies the Lipschitz condition in Lemma 1. \newline
\textit{Assumption 5:} The pair (A,C) in \eqref{sys} and \eqref{op} is observable \cite{40}.\par \vspace{0.1cm}

Assumption 1 is a bounded input bounded output (BIBO) stability condition which is a practical consideration because the output, $V(t)$ of Li-ion cell can never be unbounded for a bounded input current, $I(t)$. The knowledge about various uncertainties can be determined based on experiments and domain knowledge. Thus Assumption 2 allows us to consider the effect of measurement noise in the state and modelling inaccuracies separately in the observer design. Since the nonlinear map between $V_{OC}(z(t))$ and $z(t)$ is continuous and monotonic, Assumption 4 is also valid for the current problem and can be verified.

\begin{theorem}
For the system (\ref{sys}) satisfying  the uncertainty bounds in Assumptions (1)-(4), if the observer \eqref{ob} with gain $L_s=\frac{\mu}{2}P^{-1}C^T$ and observer gain matrix $L$ fullfills the following matrix inequality,

\begin{equation}
    \tilde{W}\left( P,L\mid \alpha ,\varepsilon \right)=\left[ 
\begin{array}{cc}
\Xi & P \\ 
P & -\varepsilon I_{n\times n}
\end{array}
\right] <0
\end{equation}
where $\Xi\triangleq P(A-LC+\frac{\alpha}{2}I_{n \times n})+ (A-LC+\frac{\alpha}{2}I_{n \times n})^TP+\varepsilon L_{\phi}^2I_{n\times n}$, for some positive definite symmetric matrix $P=P^{T}>0$ and positive constants $\alpha>0 ,\varepsilon>0 $ and $\mu > 0$, then the state estimation error $\tilde x(t)$ converges to a bounded region
    $\|\tilde x(t)\|^2 \leq \frac{1}{\lambda_{min}(P)} \left( \frac{c}{\alpha }+O\left(e^{-\alpha t}\right) \right)$, where $c\triangleq \varepsilon \gamma^2 X_{+}^2+4D_+^2$ and $O(e^{-\alpha t})\triangleq {\tilde x}_0^TP{\tilde x}_0~ e^{-\alpha t}-\frac{c}{\alpha}e^{-\alpha t}$, where ${\tilde x}_0$ is the initial estimation error.
\end{theorem}

\begin{proof}
Let us consider a Lyapunov candidate function 
\begin{equation}\label{lyap}
    V(\tilde x)=\tilde x^T P \tilde x
\end{equation}
where $V$ is continuously differentiable, positive definite and radially unbounded and $P$ is a symmetric positive definite matrix \cite{40}. Taking the time derivative of \eqref{lyap} and using \eqref{err} yeilds
\begin{equation}\label{pr4}
    \dot V(\tilde x) = 2\tilde x^T P(A-LC)\tilde x
    + 2 \tilde x^T P\left( \underbrace{\Delta A x +\Delta \phi+Dd(t)}_{\xi}\right) \\
    - 2 \tilde x^T PL_s SIGN(\sigma)
\end{equation}
Choosing $L_s=\frac{\mu}{2}P^{-1}C^T$ and using the relation $ \sum_{i=1}^{m} \arrowvert \sigma_i \arrowvert \geq \|\sigma\|_2$, we can write
\begin{equation}\label{pr3}
     2 \tilde x^T PL_s SIGN(\sigma) \geq \mu \|\sigma\|
\end{equation}
\hspace{-0.1cm}Upper bounding \eqref{pr4} using \eqref{pr3}
\begin{equation}\label{eqn}
    \dot V\leq 2 \tilde x^T P(A-LC)\tilde x+2{\tilde x}^TP \xi -\mu \|\sigma\|
\end{equation}
where $\xi \triangleq \Delta Ax+\Delta \phi+Dd(t)$. Now, expressing \eqref{eqn} into symmetric form and adding and subtracting $\varepsilon I_{n \times n}$ and $\alpha V(\tilde x)$ with scalars $\varepsilon,\alpha>0$ on the left side of (\ref{pr4})
\begin{equation}\label{pr5}
    \dot V(\tilde x) = \left(
    \begin{array}{c}
\tilde x\left( t\right) \\ 
\xi \left( t\right)%
\end{array}
\right) ^{T}W  \left(
    \begin{array}{c}
\tilde x\left( t\right) \\ 
\xi \left( t\right)%
\end{array}
\right)\\
+\varepsilon \left\Vert \xi  \right\Vert ^{2}-\alpha
V\left(\tilde x\right) -\mu \|\sigma\|
\end{equation}
where $W \triangleq \left[ 
\begin{array}{cc}
P(A-LC)+(A-LC)^TP+\alpha P & P \\ 
P & -\varepsilon I_{n\times n}%
\end{array}%
\right]$. Now expanding $\|\xi\|^2 = \|\Delta A x + \Delta \phi + D d(t) \|^2$ in \eqref{pr5} and using the bounds in Assumptions (1)-(4)
\begin{eqnarray}\label{pr6}
\|\xi\|^2&=&\|\Delta A x + \Delta \phi + D d(t) \|^2\\ \nonumber
&\leq& \|\Delta A x\|^2 + \|\Delta \phi\|^2 + \|Dd(t)\|^2\\\nonumber
&\leq&\gamma^2 X_{+}^2 + L_{\phi}^2 \|\tilde x\|^2 + 4D_{+}^2 \nonumber
\end{eqnarray}
Substituting \eqref{pr6} in \eqref{pr5}, we get
\begin{equation}\label{pr7}
    \dot V(\tilde x) \leq \left(
    \begin{array}{c}
\tilde x\left( t\right) \\ 
\xi \left( t\right)%
\end{array}
\right) ^{T}
\underbrace{
\left[ 
\begin{array}{cc}
\Xi & P \\ 
P & -\varepsilon I_{n\times n}%
\end{array}%
\right]}
_{\tilde{W}\left( P,L\mid \alpha ,\varepsilon \right)}
\left( 
\begin{array}{c}
\tilde x\left( t\right) \\ 
\xi \left( t\right)%
\end{array}%
\right)-\alpha V\left(\tilde x \right) 
+\underbrace{\varepsilon \gamma^2 X_{+}^2 + 4D_{+}^2}_{c}
-\mu \|\sigma\|
\end{equation}
Now, if $\tilde{W}\left( P,L\mid \alpha ,\varepsilon \right) < 0$, then from \eqref{pr7}, 
\begin{equation}\label{pr8}
    \dot V(\tilde x) \leq -\alpha V(\tilde x)+c-\mu\|\sigma\|
\end{equation}
where $c>0$ is a positive scalar constant that depends on the bounds on uncertainty. Equation \eqref{pr8} can be further upper bounded as
\begin{equation}\label{eqnn}
    \dot V(\tilde x)\leq -\alpha V(\tilde x)+c
\end{equation}
The solution of \eqref{eqnn} can be obtained as
\begin{equation}\label{eqnnn}
    V(\tilde x)\leq V({\tilde x}_0)e^{-\alpha t} + \frac{c}{\alpha}(1-e^{-\alpha t})
\end{equation}

From \eqref{eqnnn}, one can obtain
\begin{equation}\label{limsup}
  {\lim \sup}_{t\rightarrow\infty} V\left( \tilde x\left( t\right) \right)\leq \frac{c}{\alpha}. 
\end{equation}
Further (\ref{limsup}) can be equivalently written as, 
\begin{equation}\label{eq17}
    {\lim \sup}_{t\rightarrow\infty} \tilde x^{T }\left( t\right) \left[
P_{attr}\right] \tilde x\left( t\right)\leq 1
\end{equation}
where $P_{attr}:=\frac{\alpha }{c}P$ is the ellipsoidal matrix. Hence the stability of the state estimation error dynamics is proved since the time derivative of the storage function $V(\tilde x)$ is uniformly ultimately bounded (UUB) under bounded uncertainty and disturbance. 
\end{proof}
\textit{Definition:} An ellipsoidal set $E(P_{attr}) \triangleq \{\tilde x|{\tilde x}^TP_{attr}\tilde x \leq 1 \}$ 
where the ellipsoidal matrix $P_{attr}$ is a symmetric positive definite matrix. $E(P_{attr})$ is called as an attractive ellipsoid fo the system \eqref{err} if it is a globally asymptotic attractive invariant set \cite{40}.\par
 It is to be worth mentioning at this stage that the size of $E(P_{attr})$ is minimized by solving a convex optimization problem with matrix constraints which is presented in the corollary as provided below.
\begin{corollary}\cite{411,42,43}
The optimal parameter $L^{\ast}$ for the proposed observer is computed by solving a semidefinite programming problem (SDP) as follows:
\begin{equation}\label{convex}
     \underset{P>0,L,\alpha >0,\varepsilon >0}{\text{minimize}} tr\left(P_{attr}\right)
\end{equation}
where the operator $\mathrm{tr}(\cdot)$ represents the trace operator that is operated on the matrix $P_{attr}$ satisfying the following matrix inequality,
\begin{equation}\label{lmi_ob1}
\tilde{W}\left( P_{attr},L\mid \alpha ,\varepsilon \right) <0 
\end{equation}
\end{corollary} 
The bilinear matrix inequality (BMI) in (\ref{lmi_ob1}) is essentially nonlinear which needs to be converted to an LMI for solving the above convex optimization problem using standard LMI solvers. 
\begin{corollary}\cite{411}
The inequality in (\ref{lmi_ob1}) can be converted to LMI if the first element of $\tilde{W}\left( P,L\mid \alpha,\varepsilon \right)$, i.e., $\Xi$ is modified  by introducing a new variable $Y\triangleq PL$ as follows:
\begin{equation}\label{LMI}
\tilde{W}\left( P_{attr},L^{\ast}\mid \alpha ,\varepsilon \right)=
\left[ 
\begin{array}{cc}
\Xi & P \\ 
P & -\varepsilon I_{n\times n}
\end{array}
\right]<0,~~P>0
\end{equation}
where $\Xi=PA-YC+ A^TP-Y^TP+\alpha P+\varepsilon L_{\phi}^2I_{n\times n}$ and the optimal observer gain matrix, $L^{\ast}$ is computed as
\begin{equation}\label{o_gain}
    L^{\ast}=P_{attr}^{\ast^{-1}}Y
\end{equation}
\end{corollary}
For fixed parameters, $\alpha$ and $\varepsilon$, the matrix inequalities presented in \eqref{LMI} become linear which can be solved using MATLAB toolboxes SeDuMi and YALMIP \cite{44}.
\section{RESULTS AND DISCUSSION}
In this section experimental studies and numerical simulations are reported to validate and evaluate the proposed SoC estimation technique. At first, experimental validation of the real-time performance of the proposed AESMO is performed through HPPC and UDDS tests. Next, a comparative study with an existing robust observer is done through simulations. Subsequently, the effect of measurement noise in the current and voltage channels of the experimental system on the observer's performance is investigated. Finally, the effect of ageing in terms of variation of the internal resistance is analyzed through Monte Carlo simulations.\par
The value of the model parameters of the system in \eqref{sys} is chosen to correspond to SoC, $z(t)=0.1$ as provided in Table \ref{tab1}. It is important to note that the identified parameters of an unused Li-ion cell can not accurately represent the dynamic behaviour of the cell after repeated cycles of usage. Here, the model parameters of the ECM model in \eqref{sys} are kept constant throughout the designed experiments and simulations (irrespective of the SoC value at different operating conditions) to account for the above mentioned uncertainty. For $\mu=10^{-10}$ and a Lipschitz constant, $L_{\phi}=0.8$ choosing the design parameters as $\alpha=2\times 10^{7}$ and $\varepsilon=2\times 10^{-8}$, the SDP problem in \eqref{convex} with LMI constraints in \eqref{LMI} is solved which provided an optimal observer gain matrix, $L=[0.3645~-0.2364~2.002\times 10^{-8}~0.0217]^T$. It can be noted that the value of the observer gain is considered to be the same for all experiments and simulations.

\subsection{Experimental Validation of the SoC Estimation Technique}
\begin{figure}
\centering
	\includegraphics[width=8.5cm]{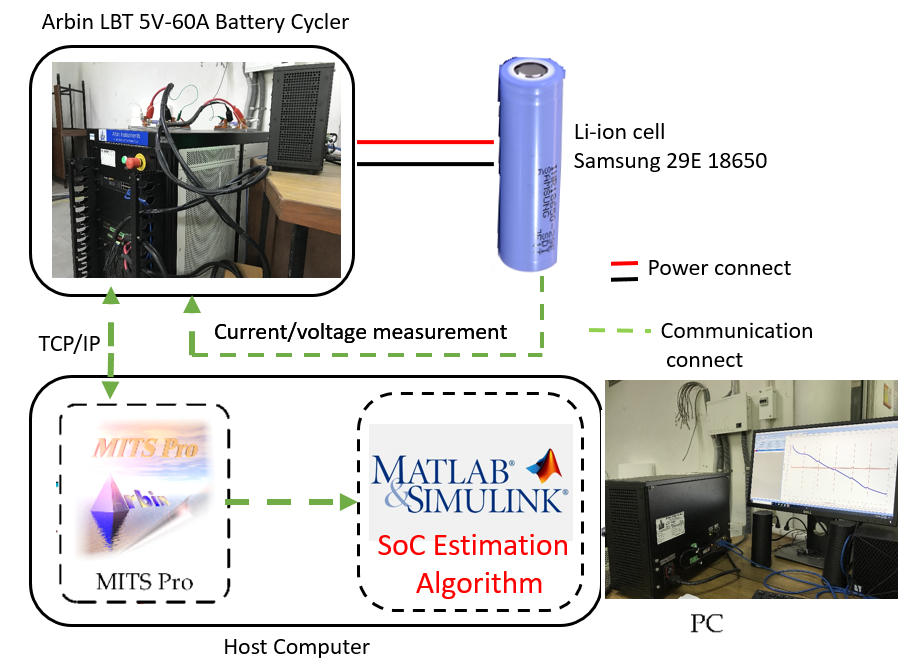}
	\caption{Schematic diagram of the battery test bench.}\label{FIG_exp}
\end{figure}
As shown in Fig. \ref{FIG_exp} high precision battery testing system, Arbin LBT21084 5V-60A Battery Cycler is utilized to carry out experiments to determine the characteristics of Li-ion cells, identify the ECM parameters of the cell, and quantify the efficacy of the proposed AESMO based SoC estimation technique. For the designed experiments Li-ion cells INR18650-29E made by Samsung SDI Co., Ltd., are utilized having a nominal voltage and nominal capacity of 3.65 V and 2.85 Ah, respectively. The current and voltage measurement data are stored in the memory of the host computer with a sampling rate of $1$ Hz. Finally, the recorded current and voltage data is imported to MATLAB/SIMULINK to be used in the SoC estimation algorithm.\par
\subsubsection{HPPC Test} 
To investigate the effectiveness of the proposed AESMO, a new HPPC test is designed with a larger number of the current discharging pulses. The corresponding current and voltage profile of the new HPPC test is depicted in Fig. \ref{FIG_hppc1}. The HPPC experimental test data obtained from the experimental setup is utilized for evaluating the efficacy of the proposed AESMO. The objective of this experiment is to investigate whether the proposed observer can provide an acceptable accuracy level of $\pm5\%$ in SoC estimation in different regions of SoC while keeping the ECM parameters fixed (as discussed before) during the entire experiment. 

\begin{figure}
\centering
\subfigure[]{
\hspace*{-0.5cm}
  \includegraphics[width=75mm,height=4.5cm]{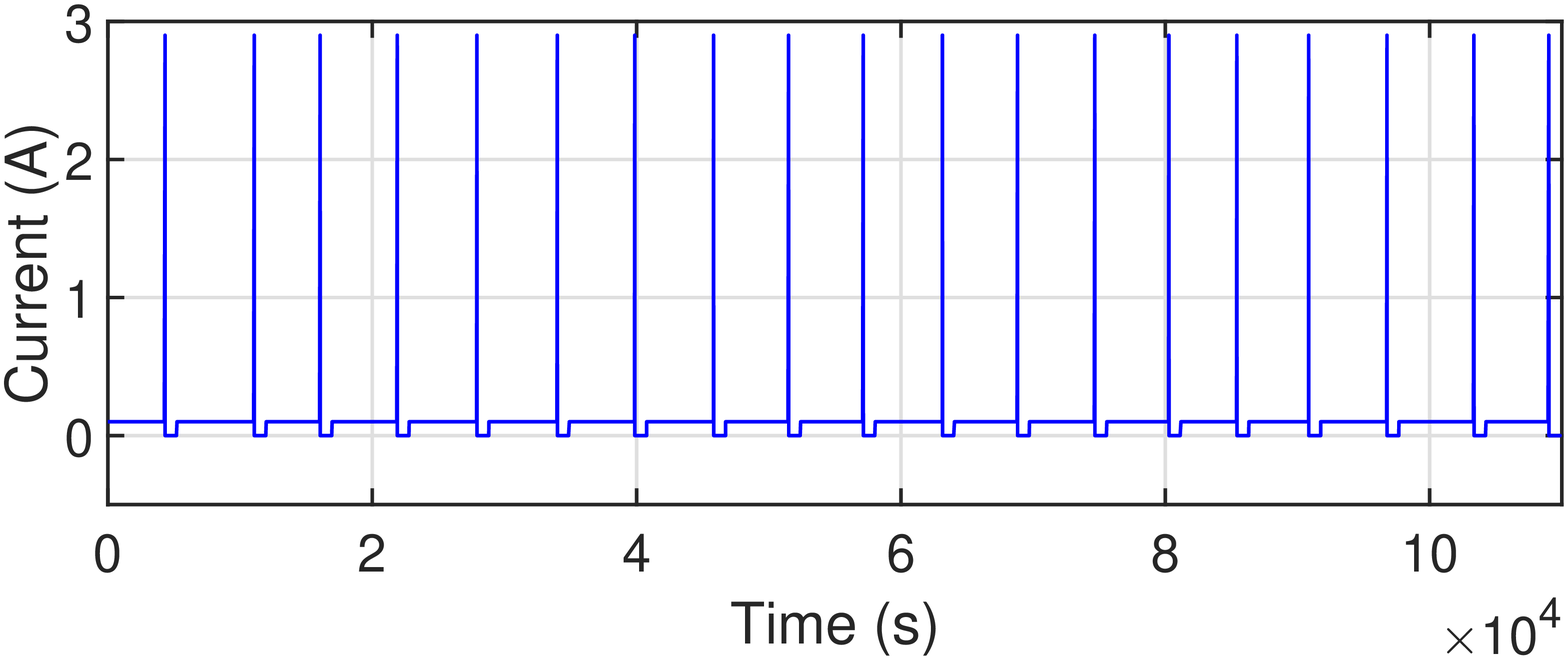}
}
\subfigure[]{
  \includegraphics[width=75mm,height=4.5cm]{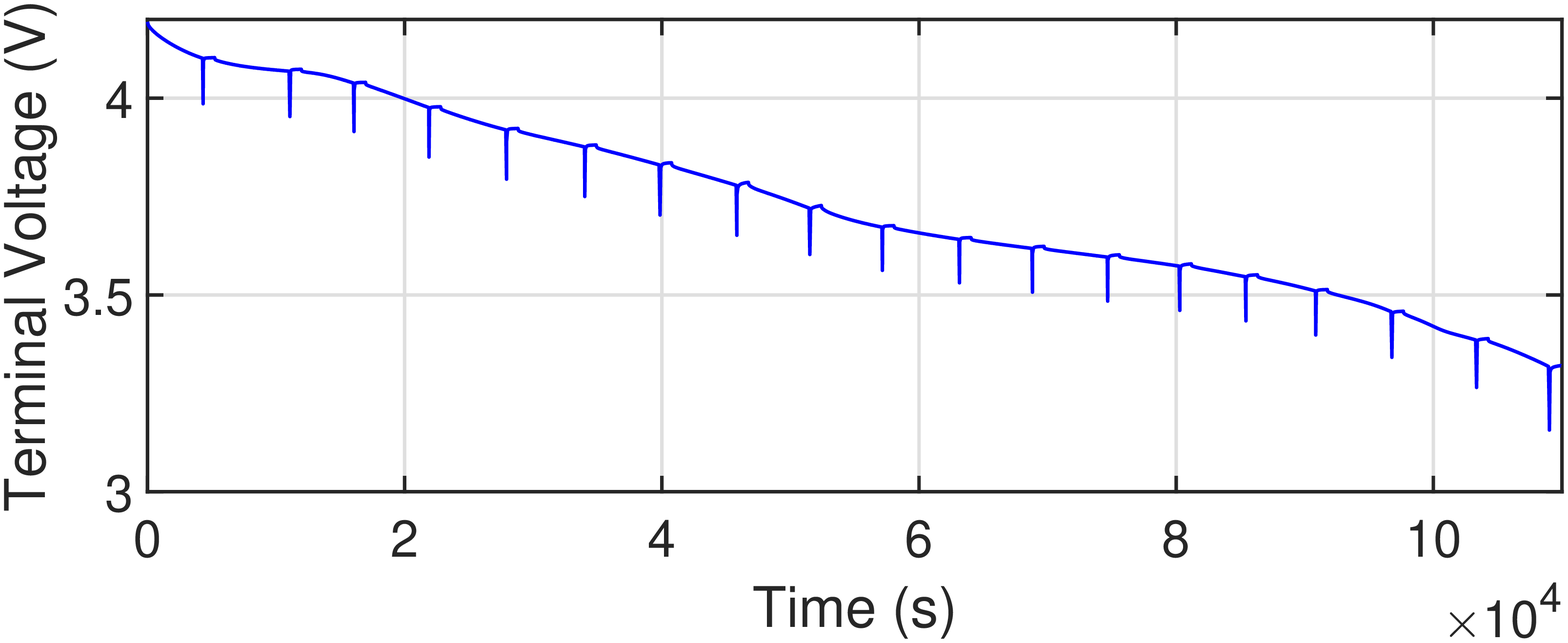}
}
\caption{Measured current and terminal voltage of HPPC test for the Li-ion cell.}
\label{FIG_hppc1}
\end{figure}

\begin{figure}
\centering
\subfigure[]{
\hspace*{-0.5cm}
  \includegraphics[width=75mm,height=4.5cm]{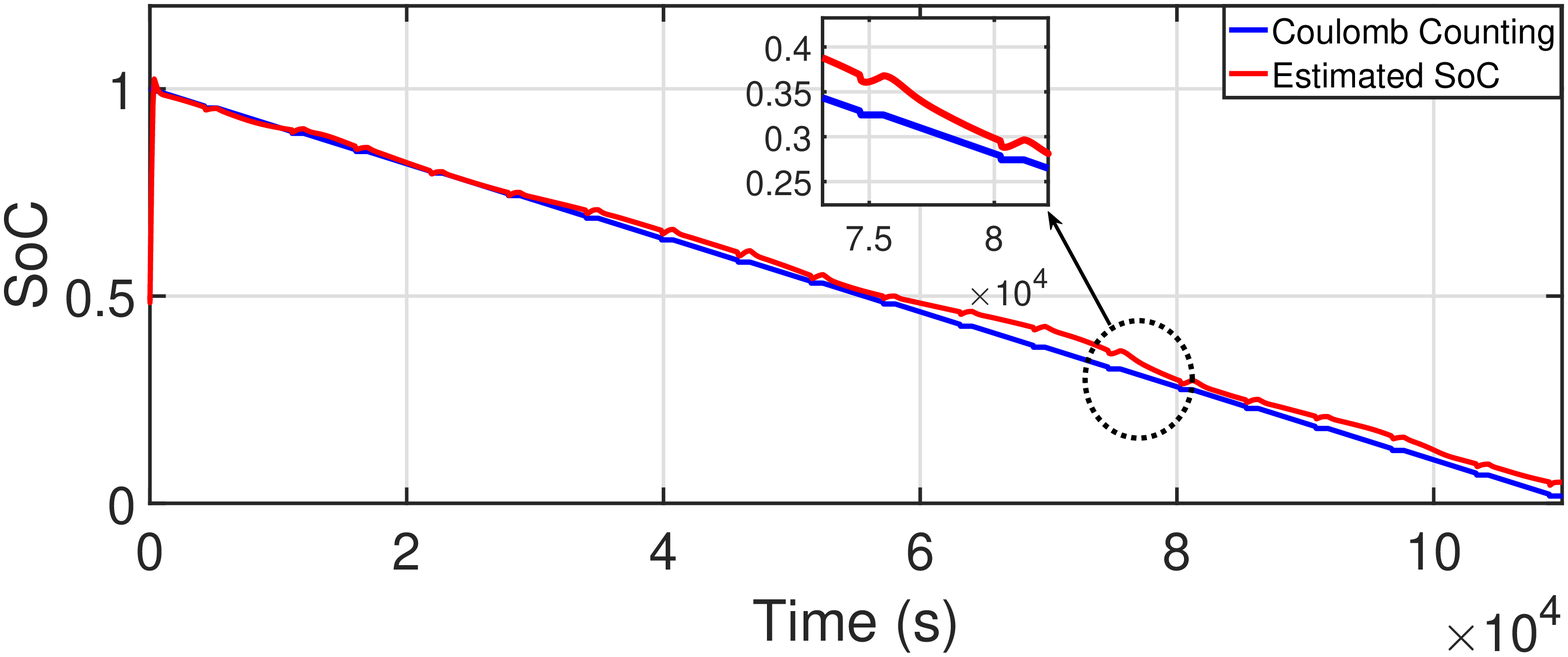}
}
\subfigure[]{
  \includegraphics[width=75mm,height=4.5cm]{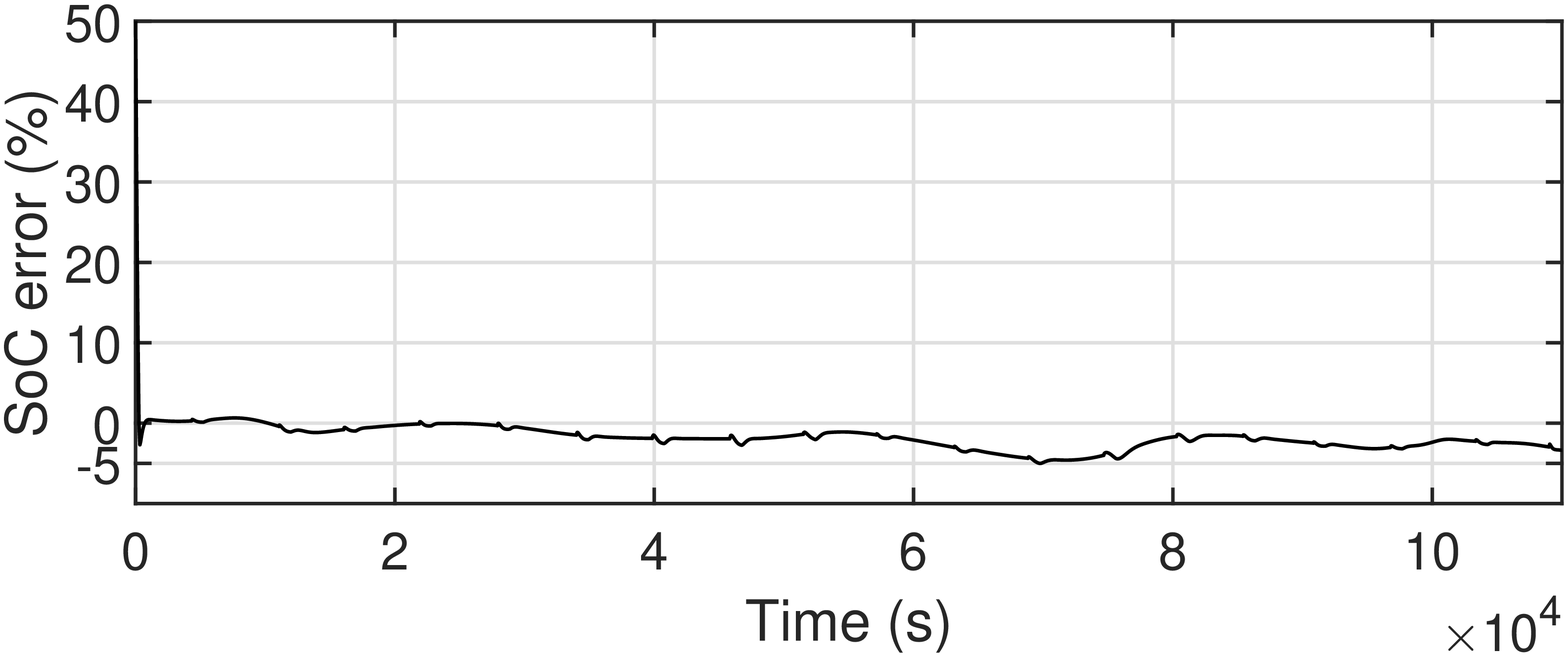}
}
\caption{SoC estimation results for HPPC test.}
\label{FIG_HPPCEXP}
\end{figure}
As illustrated in Fig. \ref{FIG_HPPCEXP}(a), one can observe that despite the error in initial SoC estimate ($z_0(t)=0.6$), the proposed AESMO can estimate the experimental SoC satisfactorily. The corresponding percentage error in SoC estimation for this case is restricted between $\pm5\%$ which validates the above fact as depicted in Fig. \ref{FIG_HPPCEXP}(b).

\subsubsection{UDDS Test:} 
\begin{figure}
\centering
\subfigure[]{
\hspace*{-0.5cm}
  \includegraphics[width=75mm,height=4.5cm]{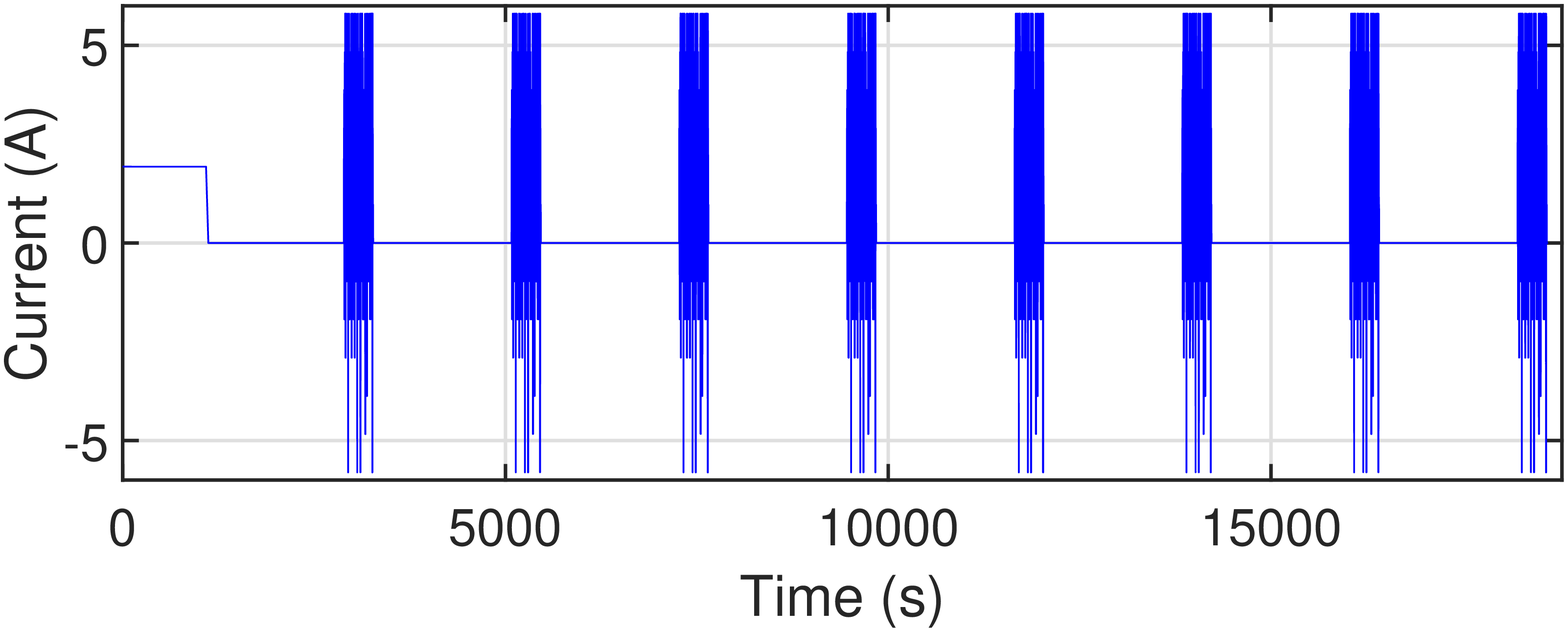}
}
\subfigure[]{
  \includegraphics[width=75mm,height=4.5cm]{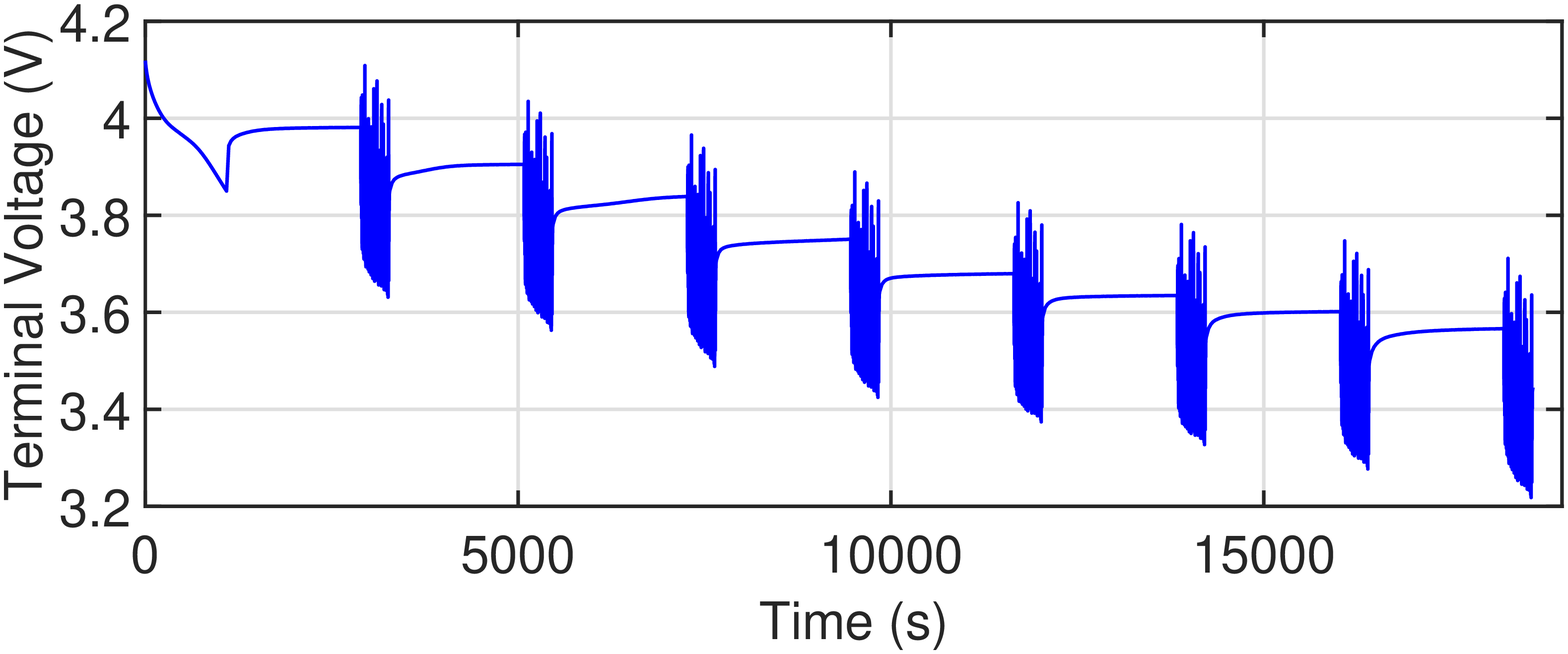}
}
\caption{Input current and terminal voltage of UDDS test for the Li-ion cell.}
\label{FIG_UDDS}
\end{figure}

The dynamic current profile of UDDS test as depicted in Fig. \ref{FIG_UDDS}(a) is considered for the validation of the proposed observer. It represents the effect of high-current excitation on the performance of the Li-ion cell. Since the ECM parameters are derived from a slow current discharge profile, the UDDS experiment is crucial to determine the effectiveness of the proposed method. The same assumption on the ECM parameter is considered for this experiment. The corresponding terminal voltage profile of the Li-ion cell is illustrated in Fig. \ref{FIG_UDDS}(b).

\begin{figure}
\centering
\subfigure[]{
\hspace*{-0.5cm}
  \includegraphics[width=75mm,height=4.5cm]{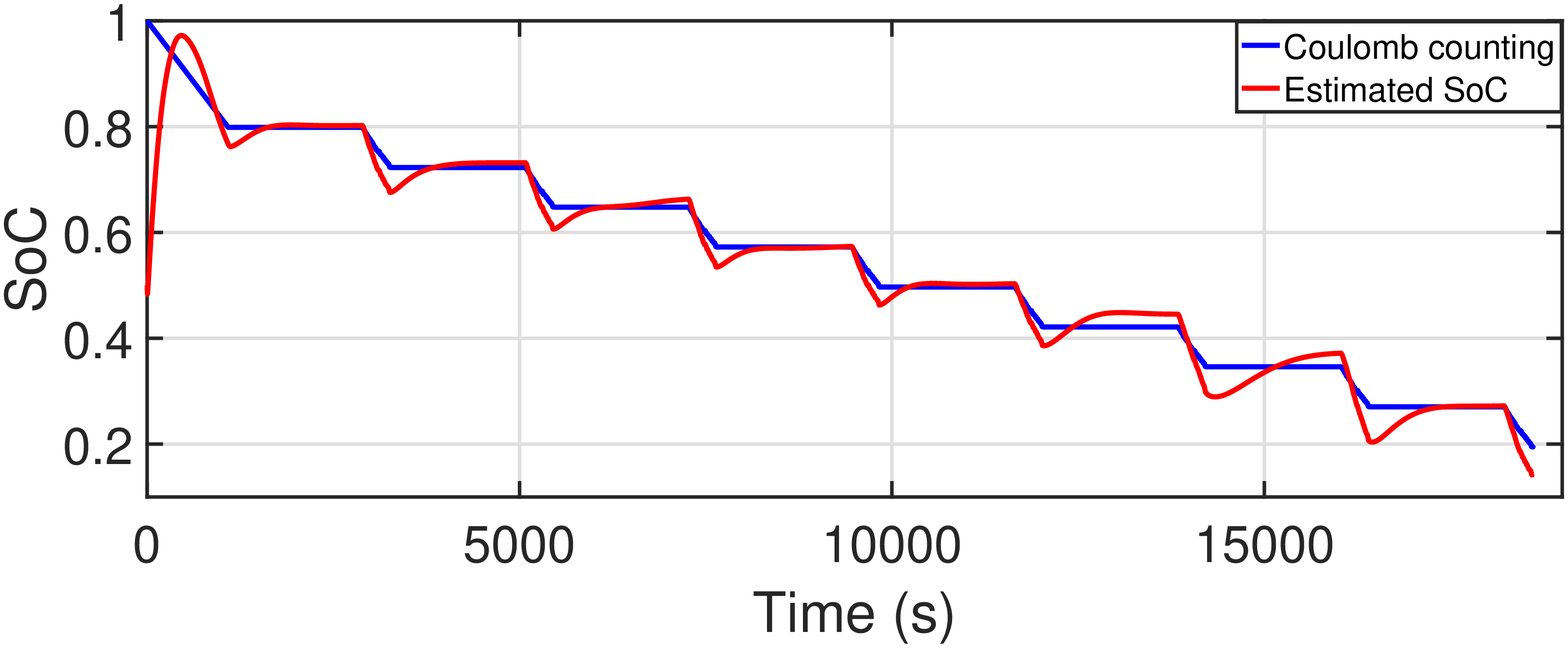}
}
\subfigure[]{
  \includegraphics[width=75mm,height=4.5cm]{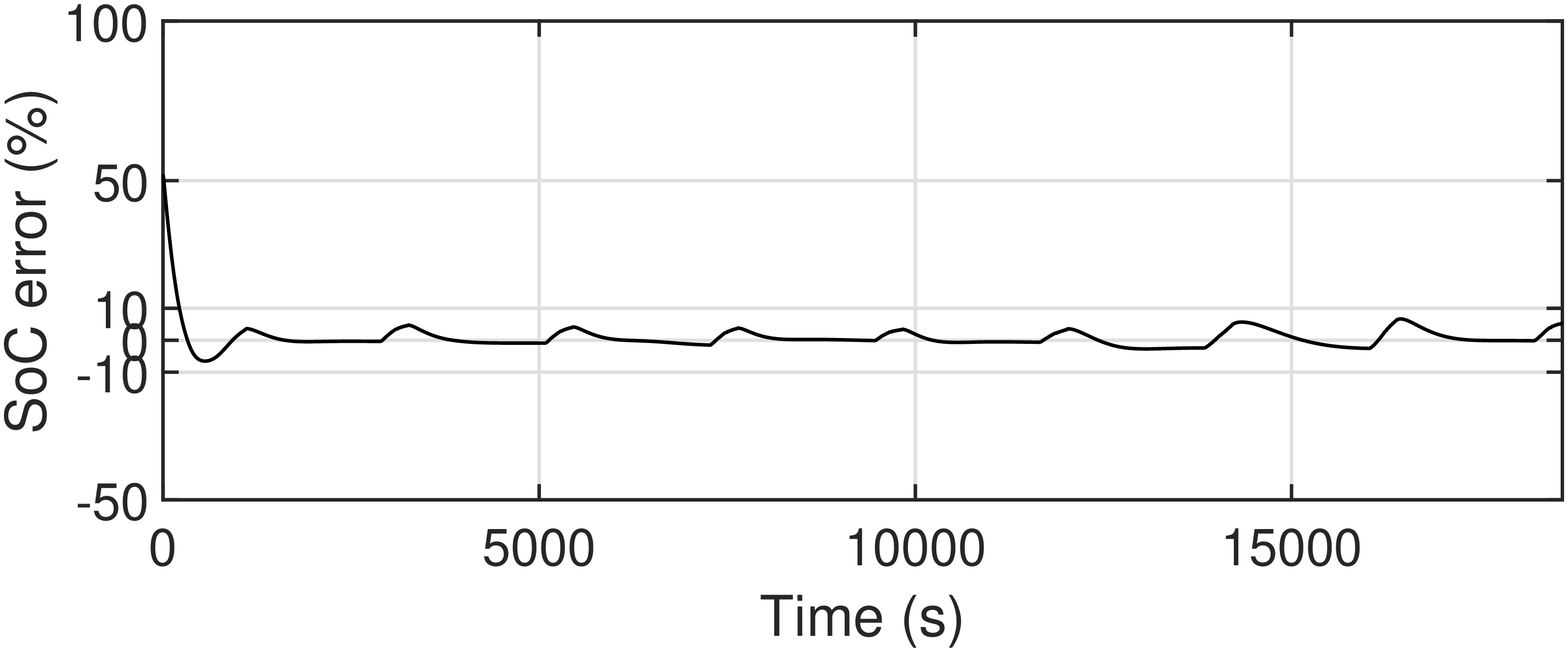}
}
\caption{SoC estimation results for UDDS test.}
\label{FIG_UDDSEXP}
\end{figure}
The experimental results of the SoC estimation are shown in Fig. \ref{FIG_UDDSEXP}(a). It can be observed that the proposed observer, starting from the erroneous initial value converges to the profile of SoC of the Li-ion cell. The corresponding percentage SoC estimation error is confined within the $\pm10\%$ error bound during the entire experiment as shown in Fig. \ref{FIG_UDDSEXP}(b). The result validates the effectiveness of the proposed SoC estimation technique for different operating conditions. It is important to note that percentage of error in SoC estimation for both the HPPC and UDDS tests are relatively higher when SoC is more than 0.5 as illustrated in Fig. \ref{FIG_HPPCEXP}(b) and Fig. \ref{FIG_UDDSEXP}(b). The reason behind it is that the parameters of the ECM corresponds to values identified at SoC=0.1. As a result, these parameters produce significant deviation of the SoC error for SoC more than 0.5.

\begin{table}[!h]
\renewcommand{\arraystretch}{0.1}
\centering \caption{Performance evaluation of the proposed observer.}\vspace{0.1cm}
\label{tab2} 
\begin{tabular}{cccc}
\hline 
\hline
\vspace{0.1cm}
S. No & Experiment & IAE & ISE\\ 
\hline \vspace{0.1cm}
$1$ & HPPC & $0.03351$ & $0.002733$\\ \vspace{0.1cm}
$2$ & UDDS & $0.05227$ & $0.001123$\\ 
\hline
\hline
\end{tabular}
\end{table}

The performance of the proposed AESMO in terms of the evaluation of the SoC estimate in the experiments as discussed before are summarized in terms of performance indices, such as integral absolute error (IAE) and integral square error (ISE) as provided in Table \ref{tab2}. The following inference can be drawn from Table \ref{tab2}: (i) a low value of IAE validates that the overall error in SoC estimation is small, (ii) similarly, a low value of ISE indicates that less estimation error can persist in the design. 

\subsection{Simulation Results}
To further illustrate the advantages of the proposed AESMO, it is compared with an existing robust observer algorithm as reported in \cite{11}. To incorporate modelling inaccuracies and exogenous disturbances, a bounded sinusoidal disturbance amplitude of $\pm5\%$ and frequency of $0.27$ mHz, respectively, is considered. The observer gain for the system in \eqref{sys} using the design method in \cite{11} is computed as $L=[0.0288~-0.0032~3.282\times 10^{-9}~-9.556\times 10^{-5}]^T$. Similarly, the performance of AESMO with an UKF-based SoC estimation technique as in \cite{19} is also investigated here. The UKF is designed for the system with $z(t)$, $V_{RC_1}(t)$ and $V_{RC_2}(t)$ as state variables as provided in \eqref{eq2}-\eqref{eq4}  and $V(t)$ as output in \eqref{eq1}. The parameters of UKF are considered as follows: (i) the initial covariance matrix as diag $(1\times 10^{-12},1\times 10^{-8},1)$, where $diag(a,b,\dots)$ denotes diagonal matrix with elements $a,~b,\dots$, (ii) the additive process noise as diag$(1\times 10^{-8},1\times 10^{-8},1\times 10^{-6})$ and (iii) the unscented transformation parameters as $\alpha=1$, $\beta=2$ and $\kappa=0$. It is important to note that the simulation setting is identical for all the above mentioned SoC estimation techniques. Fig. \ref{FIG_hppccomp} depicts the comparison results of the estimation of $z(t)$ of the Li-ion cell with a highly underestimated guess for initial SoC value under the HPPC test.\newline

\begin{figure}
\centering
	\includegraphics[width=8.5cm]{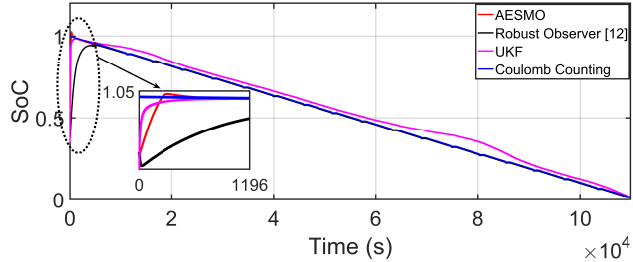}
	\caption{Comparative study of SoC estimation with the proposed AESMO, robust observer in \cite{11} and UKF.}\label{FIG_hppccomp}
\end{figure}
As illustrated in Fig. \ref{FIG_hppccomp}, both the proposed AESMO as well as the robust observer in \cite{11} successfully provide an accurate estimation of the true SoC of the system with high accuracy in the presence of modelling inaccuracies. But the UKF-based algorithm is less accurate as compared to the observer-based algorithms since it requires a precise model for SoC estimation. Moreover, it is evident that the proposed AESMO outperformed the robust observer in \cite{11} in terms of faster convergence. This is due to the design parameter $\mu$ in the observer gain matrix of the sliding mode term in the \eqref{ob}. The parameter $\mu$ provides an additional design advantage in terms of a faster convergence of the observer as compared to \cite{11}.

\subsection{SoC estimation with both current and voltage noise:} Unlike the precise battery testing equipment utilized in the previous experiments, significant noise can exist in the current and voltage measurement channels simultaneously in practical battery systems. Thus, a new experiment is designed with both current and voltage zero mean Gaussian noise signals in this paper. Fig. \ref{FIG_NOISE} illustrates the performance of SoC estimation of the proposed AESMO where $5\%$ current noise and $1\%$ voltage noise is added to the HPPC test data as in \cite{45}. It can be observed that the SoC estimation error stays within $5\%$ which validates the robustness of the proposed AESMO to measurement noise. 

\begin{figure}
\centering
	\includegraphics[width=8.5cm,height=4.5cm]{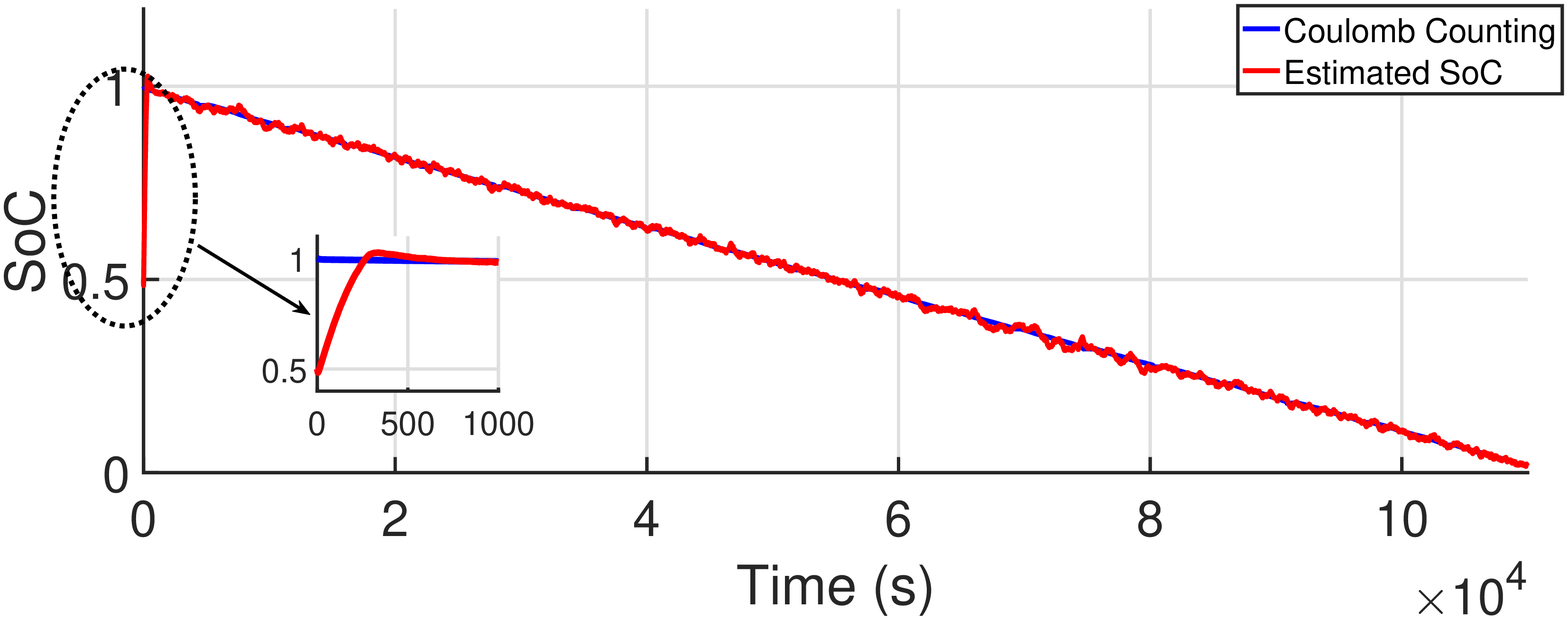}
	\caption{Effect of noise in current and voltage measurement.}\label{FIG_NOISE}
\end{figure}

\subsection{Effect of a large variation in internal resistance of Li-ion cell} The following simulation study is designed to investigate the robust performance of the proposed observer against $\pm20\%$ variations in the internal resistance of the Li-ion cell. It is a very practical consideration since ageing of a cell can drastically change the internal resistance, $R_{int}$, due to repeated usage. 

\begin{figure}
\centering
	\includegraphics[width=8.5cm]{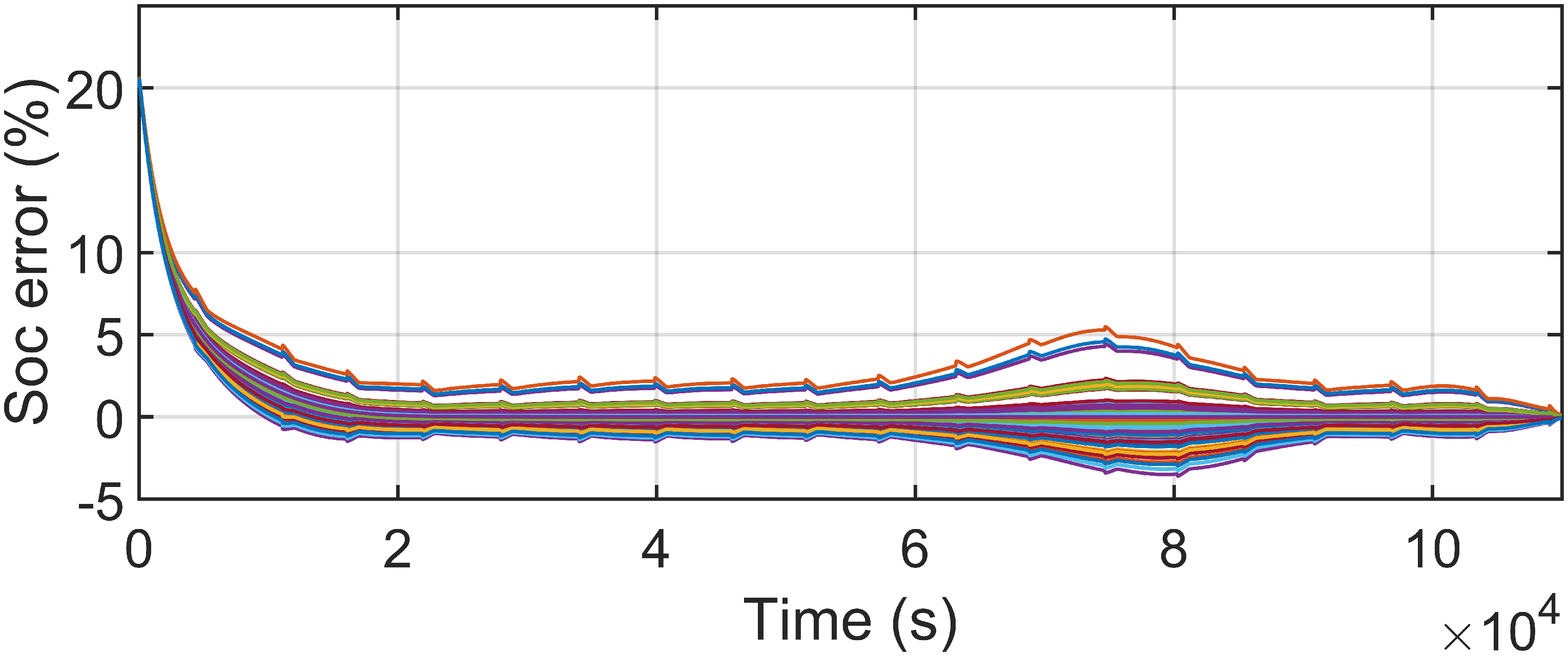}
	\caption{Effect of noise in current and voltage measurement.}\label{FIG_R0}
\end{figure}
 It can be observed from Fig. \ref{FIG_R0} that the percentage SoC estimation error is confined to $\pm5\%$ error band which indicates that the ECM parameter need not be updated when the internal resistance changes by up to $\pm20\%$ after repeated cycling. 
\section{Conclusion}
In this paper, a robust attractive ellipsoid based SMO observer design is proposed for Li-ion cells that provides robust estimates of the SoC in the presence of bounded parametric uncertainty, modelling inaccuracies and exogenous disturbance. Hence, in an uncertain environment, the proposed algorithm can provide reliable SoC estimates. The observer gain was shown to be an optimal feedback parameter for the proposed observer since it is numerically computed by solving a semi-definite programming problem. Furthermore, a faster convergence rate of the observer states is attainable by tuning the design parameter. Due to its simplistic design, it can be easily implemented on low-cost devices. Extensive numerical simulations and experimental studies validate the capability of the proposed observer in estimating SoC in real-time in the presence of noise in the current and voltage measurements, modelling inaccuracies and wide range of parametric variations. \par

Though only a single Li-ion cell is considered, the extension of the proposed technique to the management of a battery pack can be an immediate future work. The accurate state estimation of battery packs is still an open problem due to inconsistent battery pack characteristics and uncertain operating conditions. The proposed technique can be augmented with the existing techniques, such as cell calculation, screening process, and bias correction methods for efficient state estimation of battery packs. Apart from the ECM based methods, the proposed technique can be applied to physics-based models of Li-ion cells in the future.

\bibliographystyle{unsrt}  
\bibliography{references}  






\end{document}